\documentclass{eptcs}
\providecommand{\event}{DICE 2011} 
\usepackage{amsmath,amssymb,amsthm,proof}
\newcommand{\itn}{\mathbf{IT}(\mathbb{N})}
\newcommand{\itnd}{\overline{\mathbf{IT}}(\mathbb{N})}

\newcommand{\A}[1]{\mathbf{A}[#1]}
\newcommand{\PR}{\mathrm{PR}}

\newcommand{\proves}[1]{\vdash^{\kern -4pt\mbox{\raise 2pt \hbox{\tiny\itshape #1}}}}

\newtheorem{theorem}{Theorem}
\newtheorem{lemma}[theorem]{Lemma}

\hypersetup{colorlinks=true}
\title{Provably Total Functions of Arithmetic with Basic Terms}
\author{Evgeny Makarov
\institute{INRIA\\Orsay, France}
\email{emakarov@gmail.com}
}
\def\titlerunning{Provably Total Functions of Arithmetic with Basic Terms}
\def\authorrunning{E. Makarov}

\begin{document}
\maketitle

\begin{abstract}
A new characterization of provably recursive functions of first-order
arithmetic is described. Its main feature is using only basic terms, i.e., terms
consisting of $\mathsf{0}$, the successor $\mathsf{S}$ and variables in the quantifier
rules, namely, universal elimination and existential introduction.
\end{abstract}

\section{Introduction}
\label{sec:1}
This paper presents a new characterization of provably recursive
functions of first-order arithmetic. We consider functions defined by
sets of equations. The equations can be arbitrary, not
necessarily defining primitive recursive, or even total, functions.
The main result states that a function is provably recursive iff its
totality is provable (using natural deduction) from the defining set
of equations, with one restriction: only terms consisting of $\mathsf{0}$, the
successor $\mathsf{S}$ and variables can be used in the inference rules
dealing with quantifiers, namely universal elimination and existential
introduction. We call such terms \emph{basic}.

Provably recursive functions is a classic topic in proof
theory~\cite{buss:1998}. Let $T(e,\vec{x},y)$ be an arithmetic formula
expressing that a deterministic Turing machine with a code $e$ terminates on inputs
$\vec{x}$ producing a computation trace with code $y$.
A function $f$ is a provably
recursive function of an arithmetic theory $T$ if
\begin{equation}
\label{eq:prf}
T\vdash\forall\vec{x}\,\exists y\;T(e,\vec{x},y)
\end{equation}
for the code $e$ of some Turing machine that computes $f$. In other words,
$f$ is provably recursive if the termination of its algorithm is provable
in $T$.

The class of provably recursive functions of $T$ can serve as a measure of
$T$'s strength. For example, almost all usual functions on natural
numbers are provably recursive in Peano Arithmetic (PA). In contrast, when
induction is limited to $\Sigma_1$-formulas, the set of provably recursive
functions coincides with the set of primitive recursive
functions~\cite{buss:1998}. Studying provably recursive functions is also
useful because a function that is computable but not provably recursive in
$T$ gives rise to a true formula~\eqref{eq:prf} that is independent of
$T$.

In~\cite{leivant:2002}, Leivant proposed a characterization of provably
recursive function of PA using a formalism for reasoning about inductively
generated data called \emph{intrinsic theories}. The intrinsic theory of
natural numbers has a unary data-predicate $\mathsf{N}$, which is
supposed to mean that its argument is a natural number. Unlike PA,
intrinsic theories don't use functional symbols other than the constructors
($\mathsf{0}$ and $\mathsf{S}$ in the case of natural numbers). Thus,
provably recursive functions can be characterized using only constructors
and the data-predicate. Our result goes in the same direction by
additionally replacing the data-predicate with restrictions on quantifier rules.

A deduction system with such restrictions can be considered as a way of
reasoning about non-denoting terms. A set of equations $P$ can define
non-total functions over natural numbers, and a deduction system with regular quantifier rules has
quantified variables ranging over all, not necessarily denoting, terms. For
example, a formula $\forall x\,\exists y\;f(x)=y$ is trivially provable in
a regular system regardless of the definition of $f$: we start by
$f(x)=f(x)$, introduce the existential quantifier to get
$\exists y\;f(x)=y$ and the universal quantifier to get
$\forall x\,\exists y\;f(x)=y$. In contrast, allowing only basic eigenterms in
the quantifier rules makes quantifiers range over terms denoting
natural numbers. The main result of this paper is that the formula
$\forall x\,\exists y\;f(x)=y$ is provable with this restriction iff $f$ is provably recursive.
One direction is proved using intrinsic theories; the other is proved
directly, but also following the reasoning of a similar statement
in~\cite{leivant:2002}.

The structure of the paper is the following. In the next section, relevant
definitions are given. Sect.~\ref{sec:3} shows that provably recursive
functions of PA are provably total when quantifier rules are restricted to
basic terms, and Sect.~\ref{sec:4} proves the converse.


\section{Definitions}
\label{sec:2}
Let $P$ be a set of first-order equations. Let ${\cal L}$ be the language
of $P$ plus a constant $\mathsf{0}$ and a unary functional symbol
$\mathsf{S}$ (if they are not already used in $P$). The theory $\A{P}$ is a
first-order theory with equality in the language ${\cal L}$. The axioms of
$\A{P}$ are the universal closures of the equations in $P$, denoted by
$\forall P$, the separation axioms $\forall x\;\mathsf{S}(x)\ne\mathsf{0}$
and $\forall x,y\;\mathsf{S}(x)=\mathsf{S}(y)\to x=y$, and induction
\[
A[\mathsf{0}]\to\forall x\;(A[x]\to A[\mathsf{S}(x)])\to\forall x\;A[x]
\]
for all formulas $A$ in ${\cal L}$. The inference rules are the usual rules
of classical natural deduction (see, e.g., \cite{troelstra:2000}) plus the
rules of equality:
\[
\frac{A[t]\quad t=s}{A[s]}\qquad\frac{}{t=t}
\]
for all formulas $A$ and terms $t,s$ in ${\cal L}$ ($A[s]$ is
obtained from $A[t]$ by replacing some occurrences of $t$ by
$s$). The natural deduction rules dealing with quantifiers are shown in
Fig.~1. It is easy to see that the rules of equality make it a congruence.

For example, let $\text{AM}$ be the usual axioms for addition
and multiplication and let $\PR$ be the set of standard defining
equations for all primitive recursive functions. Then $\A{\text{AM}}$
is Peano Arithmetic and $\A{\PR}$ is Peano Arithmetic with
all primitive recursive functional symbols.

\begin{figure}
\begin{gather*}
\begin{gathered}[b]
\infer[({\forall}I)]{\forall x\;A[x]}{A[y]}\\
\vtop{\hsize=3cm
  \centerline{$y$ is not free in open assumptions}}
\end{gathered}
\qquad\qquad
\begin{gathered}[b]
\infer[({\forall}E)]{A[t]}{\forall x\;A[x]}\\
\text{$t$ is free for $x$ in $A$}
\end{gathered}\\[1ex]
\begin{gathered}[b]
\infer[({\exists}I)]{\exists x\;A[x]}{A[t]}\\
\text{$t$ is free for $x$ in $A$}
\end{gathered}
\qquad\qquad
\begin{gathered}[b]
\infer[({\exists}E)]{C}{\exists x\;A[x] &
\infer*{C}{A[y]}}\\
\vtop{\hsize=3cm
\centerline{$y$ is not free in $C$}}
\end{gathered}
\end{gather*}
\caption{Quantifier rules of natural deduction}
\end{figure}

A \emph{program} is a pair $(P,\mathsf{f})$ consisting of a set of
equations $P$ and a functional symbol $\mathsf{f}$ occurring in $P$.
(When $\mathsf{f}$ is clear from the context or
is irrelevant, we will write $P$ instead of $(P,\mathsf{f})$.)

We use programs to define functions using an analog of
Herbrand-G\"odel computability (see \cite{kleene:1952,leivant:2002}).
Given a program $P$, we write $P\proves{=}E$ if $E$ is an equation
derivable from $P$ in equational logic. The rules of equational logic
are the following:
\begin{enumerate}
\item $P\proves{\tiny =}E$ for every $E\in P$;
\item $P\proves{\tiny =}t=t$ for every term $t$;
\item if $P\proves{\tiny =}E[x]$, then $P\proves{\tiny =}E[t]$ for every term $t$
  and a variable $x$;
\item if $P\proves{\tiny =}s[t]=r[t]$ and $P\proves{\tiny =}t=t'$, then
  $P\proves{\tiny =}s[t']=r[t']$.
\end{enumerate}

The relation \emph{computed by} $(P,\mathsf{f})$ is $\{(\vec{n},m) \mid
P\proves{\tiny =}\mathsf{f}(\vec{\bar{n}}) = \bar{m}\}$ (as usual,
$\bar{n}$ is a numeral for a number $n$, consisting of $n$ occurrences
of $\mathsf{S}$ applied to $\mathsf{0}$). This relation does not have to be a function.
Let us call $P$ \emph{coherent} if $P\not\proves{\tiny
  =}\bar{m}=\bar{n}$ for two distinct numerals $\bar{m}$ and
$\bar{n}$. It is easy to see that the relation computed by a coherent
program is a partial function.

However, even for a coherent program $P$ the theory $\A{P}$ can be
inconsistent because of the separation axioms. This is the case, for
example, for
$P=\{\mathsf{f}(\mathsf{g}(\mathsf{0}))=\mathsf{S}(\mathsf{g}(\mathsf{0})),
\mathsf{f}(x)=\mathsf{g}(\mathsf{0})\}$. Call a
program $P$ \emph{strongly coherent} if $\A{P}$ is consistent. It is
clear that if a program is strongly coherent, then it is coherent.

Later it will be important that a program containing a
functional symbol $\mathsf{f}$ corresponding to a primitive recursive
function $f$ also contains all defining equations for $f$.
Programs that satisfy this property are called \emph{full}.

A term is called \emph{basic} if it consists of $\mathsf{0}$, $\mathsf{S}$
and variables only. A term is called \emph{primitive recursive} if it is in
the language of $\PR$. We write
$T\proves{b}\Gamma\Rightarrow A$ (respectively,
$T\proves{pr}\Gamma\Rightarrow A$) if there is a classical natural
deduction derivation of $A$ from open assumptions $\Gamma$ in $T$ where the
eigenterms of the rules of universal elimination and existential
introduction (i.e., terms $t$ in the rules $({\forall}E)$ and
$({\exists}I)$ in Fig.~1) are basic (respectively, primitive recursive). If
$\Gamma$ is empty, we write $T\proves{b} A$ or $T\proves{pr}A$.

A function $f$ is called provable with basic terms if $f$ is computed by a
strongly coherent full program $(P,\mathsf{f})$ and
$\A{P}\proves{b}\forall\vec{x}\,\exists y\;\mathsf{f}(\vec{x})=y$, and
similarly for a function provable with primitive recursive terms.

\section{Provably recursive functions are provable with basic terms}
\label{sec:3}

In this section, we prove one direction of the main result.

\begin{lemma}
\label{lemma}
\leavevmode
\begin{enumerate}
\item $\A{\PR}\proves{b}\forall\vec{x}\,\exists y\;\mathsf{f}(\vec{x})=y$
for every functional symbol $\mathsf{f}$ from $\PR$.
\item $\A{\PR}\proves{b}\forall\vec{x}\,\exists y\;t[\vec{x}]=y$
for every primitive recursive term $t[\vec{x}]$.
\item If $\A{\PR}\vdash A$, then $\A{\PR}\proves{b}A$ for every formula $A$.
\end{enumerate}
\end{lemma}
\begin{proof}
1. By induction on the definition of the primitive recursive function $f$
corresponding to the functional symbol $\mathsf{f}$. If it is one of the base
functions, i.e., zero, addition of one or a projection, then the claim is
obvious. Suppose that $f$ is defined by composition, e.g., $\mathsf{f}(x)=\mathsf{h}(\mathsf{g}(x))$.
By induction hypothesis, we know that
\[
\A{\PR}\proves{b}\forall x\,\exists y\;\mathsf{g}(x)=y
\]
and
\begin{equation}
\label{eq:h-provable}
\A{\PR}\proves{b}\forall y\,\exists z\;\mathsf{h}(y)=z
\end{equation}
Given $x$, we can use $y$
such that $\mathsf{g}(x)=y$ to perform universal elimination on~\eqref{eq:h-provable}
and then use equality rules to derive $\exists z\;\mathsf{h}(\mathsf{g}(x))=z$ and $\exists z\;\mathsf{f}(x)=z$.

Suppose $f(\vec{x},y)$ is defined by primitive recurrence on $y$.
Then it is easy to prove $\forall y\,\exists z\;\mathsf{f}(\vec{x},y)=z$ by induction on $y$.

2. By induction on $t$, using point~1 in the induction step.

3. By induction on the derivation, using point~2 for $(\forall E)$ and $(\exists I)$.
\end{proof}

\begin{theorem}
All provably recursive functions of $\A{\PR}$ are provable with basic terms.
\end{theorem}
\begin{proof}
Suppose that $f(\vec{x})$ is provably recursive, i.e.,
$\A{\PR}\vdash\forall\vec{x}\,\exists y\;T(e,\vec{x},y)$ for some Turing
machine with code $e$ that computes $f$. It is well-known that $T$ is a
primitive recursive relation, so we can assume that $T(e,\vec{x},y)$ has
the form $\mathsf{g}(\vec{x},y)=0$ where $\mathsf{g}$ is the functional
symbol for some primitive recursive function $g$. Let $h(y)$ be the
primitive recursive function that extracts the final result from a
computation trace with code $y$. Since the machine computing $f$ is
deterministic, for each $\vec{x}$ we have $g(\vec{x},y)=0$ for exactly one
$y$.

By Lemma~\ref{lemma}.3,
$\A{\PR}\proves{b}\forall\vec{x}\,\exists
y\;\mathsf{g}(\vec{x},y)=\mathsf{0}$. Also, by Lemma~\ref{lemma}.1,
$\A{\PR}\proves{b}\forall y\,\exists z\;\mathsf{h}(y)=z$. Let $P$ be the
minimal full program containing equalities from $\PR$ for all primitive
recursive functional symbols used in these derivations, plus the
following equalities.
\begin{align*}
&\mathsf{f}(\vec{x})=\mathsf{h}(\mathsf{k}(\mathsf{g}(\vec{x},y),\vec{x},y))\\
&\mathsf{k}(\mathsf{0},\vec{x},y)=y
\end{align*}
The following is an outline of a derivation of
$\forall\vec{x}\,\exists z\;\mathsf{f}(\vec{x})=z$ in $\A{P}$. Given some $\vec{x}$, let
$y$ be such that $\mathsf{g}(\vec{x},y)=\mathsf{0}$ and let $z$ be such
that $\mathsf{h}(y)=z$. Then
$\mathsf{k}(\mathsf{g}(\vec{x},y),\vec{x},y)=y$, so
$\mathsf{f}(\vec{x})=\mathsf{h}(y)=z$.

It is left to show that $P$ is strongly coherent and computes $f$.
If $\mathsf{f}$ is interpreted by $f$ and $\mathsf{k}$ is
interpreted by the total function
\[
k(z,\vec{x},u)=
\begin{cases}
u& \text{if $z=0$,}\\
y\text{ such that }g(\vec{x},y)=0 & \text{otherwise}
\end{cases}
\]
then $\mathbb{N}\models P$; therefore,
$\A{P}$ is consistent. Further, for every $\vec{m},n$, if $f(\vec{m})=n$
then $P\proves{\tiny =}\mathsf{f}(\vec{\bar{m}})=\bar{n}$. On the other
hand, if $f(\vec{m})\ne n$, then $P\not\proves{\tiny
=}\mathsf{f}(\vec{\bar{m}})=\bar{n}$ because $f$ is total and $P$ is strongly coherent.
\end{proof}

\section{Functions that are provable with basic terms are provably recursive}
\label{sec:4}

To remind, under the assumption
$\A{P}\proves{b}\forall\vec{x}\,\exists y\;\mathsf{f}(\vec{x})=y$ we have
to prove that $f$ is provably recursive according to the definition of
Sect.~\ref{sec:1}, not that
$\A{P}\vdash\forall\vec{x}\,\exists y\;\mathsf{f}(\vec{x})=y$, which is
trivial. We will prove this statement indirectly, using intrinsic
theories~\cite{leivant:2002}.

The intrinsic theory of natural numbers, $\itn$, is a first-order
theory with equality whose vocabulary has functional symbols $\mathsf{0}$, $\mathsf{S}$
and a unary predicate symbol $\mathsf{N}$. The additional inference rules are:
\[
\frac{}{\mathsf{N}(\mathsf{0})} \qquad
\frac{\mathsf{N}(t)}{ \mathsf{N}(\mathsf{S}t)} \qquad
\frac{\mathsf{N}(t)\quad A[\mathsf{0}]\quad \forall x\;(A[x]\to A[\mathsf{S}x])}{ A[t]}\enspace.
\]
The variant of intrinsic theory that we are using, called discrete
intrinsic theory and denoted by $\itnd$ in~\cite{leivant:2002}, also
includes the separation axioms. Note that $\itnd$ uses regular first-order
quantifier rules.

A function $f$ is called provable in $\itnd$ if it is computed by a
strongly coherent program $(P,\mathsf{f})$ and
$\itnd,\forall P\vdash\forall \vec{x}\;(\mathsf{N}(\vec{x})\to \mathsf{N}(\mathsf{f}(\vec{x})))$.

The following theorem is proved in~\cite{leivant:2002}.
\begin{theorem}
\label{theorem5}
A function is provably recursive in $\A{\PR}$ iff it is provable
in $\itnd$.
\end{theorem}
Thus, it is enough to show that functions provable with basic terms are
provable in $\itnd$. In fact, we can show that functions provable with
primitive recursive terms are provable in $\itnd$.

Let us introduce some notation. If $A$ is a formula, then $A^{\mathsf{N}}$ denotes
$A$ with all quantifiers relativized to $\mathsf{N}$, i.e., having all
subformulas of the form $\forall x\;B$ replaced by $\forall
x\;(\mathsf{N}(x)\to B)$ and all subformulas of the form $\exists x\;B$
replaced by $\exists x\;(\mathsf{N}(x)\land B)$. If $\Gamma$ is a set of
formulas, then $\Gamma^{\mathsf{N}}=\{A^{\mathsf{N}}\mid A\in\Gamma\}$. If
$\vec{x}=x_1,\ldots,x_n$, then $\mathsf{N}(\vec{x})$ denotes
$\mathsf{N}(x_1)\land\ldots\land \mathsf{N}(x_n)$.

\begin{lemma}
\label{lemma6}
Let $P$ be a full program and let
$t[\vec{x}]$ be a primitive recursive term in the language of $P$. Then
$\itnd,\forall P\vdash \mathsf{N}(\vec{x})\Rightarrow \mathsf{N}(t[\vec{x}])$.
\end{lemma}
\begin{proof}
The proof is similar to Lemma~\ref{lemma}. For example, to show that a
function $f(\vec{x},y)$ defined by primitive recurrence on $y$ is provable,
one needs to use induction on the formula
$\mathsf{N}(y)\land \mathsf{N}(\mathsf{f}(\vec{x},y))$. The fullness of $P$
is necessary to ensure that the induction hypothesis is true of all
subterms of $t$.
\end{proof}

\begin{lemma}
\label{lemma7}
Suppose that $P$ is a full program and $\Gamma\cup\{A\}$ is a set of
formulas whose free variables are among $\vec{x}$. If
$\A{P}\proves{pr}\Gamma\Rightarrow A$ and all primitive recursive
functional symbols in the derivation occur in P, then
$\itnd,\forall P\vdash \mathsf{N}(\vec{x}),\Gamma^{\mathsf{N}}\Rightarrow
A^{\mathsf{N}}$.
\end{lemma}
\begin{proof}
The proof is by induction on the derivation. If $A$ is an axiom of
$\A{P}$ other than induction, then $\itnd,\forall P\vdash A$ and $A\vdash A^{\mathsf{N}}$. The only other cases
that need attention are those dealing with quantifiers and induction.

If $A[t]$ is derived from $\forall y\;A[y]$, then by induction hypothesis, $\forall
y\;(\mathsf{N}(y)\to A^{\mathsf{N}}[y])$ is derivable. Since $t$ is a primitive recursive
term in the language of $P$, $\mathsf{N}(t)$ is derivable by
Lemma~\ref{lemma6}, so $A^{\mathsf{N}}[t]$ is derivable as well. The case of
$(\exists I)$ is similar. The cases of $(\forall I)$ and $(\exists E)$
are also straightforward.

The relativized version of the induction axiom is
\[
B^{\mathsf{N}}[\mathsf{0}]\to\forall y\;(\mathsf{N}(y)\to B^{\mathsf{N}}[y]\to B^{\mathsf{N}}[\mathsf{S}y])\to\forall y\;(\mathsf{N}(y)\to
B^{\mathsf{N}}[y])\enspace.
\]
It is proved by induction in $\itnd$ for the formula $\mathsf{N}(y)\land B^{\mathsf{N}}[y]$.
\end{proof}

\begin{theorem}
All functions provable with
primitive recursive terms are provably recursive.
\end{theorem}
\begin{proof}
Let $f$ be computed by a strongly coherent full program $(P,\mathsf{f})$
and let $\A{P}\proves{pr}\forall\vec{x}\,\exists y\;\mathsf{f}(\vec{x})=y$.
Then by Lemma~\ref{lemma7},
$\itnd,\forall P\vdash\forall\vec{x}\;(\mathsf{N}(\vec{x})\to\exists
y\;\mathsf{N}(y)\land \mathsf{f}(\vec{x})=y)$.
This implies that
$\itnd,\forall P\vdash\forall\vec{x}\; (\mathsf{N}(\vec{x})\to
\mathsf{N}(\mathsf{f}(\vec{x})))$,
so by Theorem~\ref{theorem5}, $f$ is provably recursive.
\end{proof}

\subsection*{Acknowledgments}
I am grateful to Daniel Leivant, Lev Beklemishev and Tatiana Yavorskaya for
constructive discussion.

\bibliographystyle{eptcs}
\bibliography{pr}

\end{document}